\documentclass[a4paper,english, numberwithinsect]{eurocg25}

\usepackage[capitalize]{cleveref}
\usepackage{amsmath}
\usepackage{amsthm}
\usepackage{amssymb}
\usepackage{graphicx}
\usepackage{float}
\usepackage{todonotes}
\usepackage{xspace}

\crefname{conjecture}{Conjecture}{Conjectures}
\crefname{proposition}{Proposition}{Propositions}
\crefname{lemma}{Lemma}{Lemmata}
\crefname{theorem}{Theorem}{Theorems}
\crefname{section}{Section}{Sections}
\crefname{appendix}{Appendix}{Appendices}
\crefname{figure}{Figure}{Figures}
\crefname{problem}{Problem}{Problems}
\title{Flipping Matchings is Hard\footnote{Work of CB and AT supported in part by MUR PRIN Proj. 2022TS4Y3N - ``EXPAND:~scalable algorithms for EXPloratory Analyses of heterogeneous and dynamic Networked Data''. Work of FM and DP supported in part by~MUR PRIN Proj. 2022ME9Z78 - ``NextGRAAL:~Next-generation algorithms for constrained GRAph visuALization''.}}

\author[1]{Carla Binucci}%{University of Perugia, Italy}{carla.binucci@unipg.it}{}{}
\author[1]{Fabrizio Montecchiani}%{University of Perugia, Italy}{fabrizio.montecchiani@unipg.it}{}{}
\author[1]{Daniel Perz}%{University of Perugia, Italy}{daniel.perz@unipg.it}{}{}
\author[1]{Alessandra Tappini}%{University of Perugia, Italy}{alessandra.tappini@unipg.it}{}{}

\affil[1]{University of Perugia, Italy\\
\texttt{carla.binucci|fabrizio.montecchiani|daniel.perz|alessandra.tappini@unipg.it}}

\authorrunning{Binucci, Montecchiani, Perz, Tappini}

\ArticleNo{28}
\newcommand{\flipmat}{\textsc{FlippingBetweenMatchings}\xspace}
\newcommand{\planarvertexcover}{\textsc{PlanarVertexCover}\xspace}

\newcommand{\ps}{\ensuremath{\mathcal{P}}\xspace}
\newcommand{\m}{\ensuremath{\mathcal{M}}\xspace}
\newcommand{\ms}{\ensuremath{\mathcal{M}_1}\xspace}
\newcommand{\mf}{\ensuremath{\mathcal{M}_2}\xspace}

\begin{document}

\maketitle

\begin{abstract}
Given a point set \ps and a plane perfect matching \m on \ps,
a flip is an operation that replaces two edges of \m such that another plane perfect matching on \ps is obtained.
Given two plane perfect matchings on \ps, we show that it is NP-hard to minimize the number of flips that are needed to transform one matching into the other.
\end{abstract}

\section{Introduction}
A \emph{straight-line drawing} $\Gamma$ of a graph $G$ on a point set \ps maps each vertex $v$ of $G$ to a distinct point $p_v$ of \ps and each edge $(u,v)$ of $G$ to the straight-line segment $p_up_v$. % having $p_u$ and $p_v$ as its endpoints. 
Such a  drawing is \emph{planar} if no two edges share a point, except at common endpoints. In what follows, we always refer to planar straight-line drawings of graphs.

Given a point set \ps of $n$ points in the plane 
and a family $\mathcal{G}$ of  drawings, an \emph{edge flip} is the operation of replacing an edge of a drawing $\Gamma$ in $\mathcal{G}$ with a different edge such that the resulting drawing is still in $\mathcal{G}$. 
The \emph{flip graph}  of $\mathcal{G}$ is a graph that has a vertex for every element of $\mathcal{G}$ and an edge between two vertices if their corresponding drawings differ by an edge flip. The main questions about flip graphs are their connectedness, their diameter, and the complexity of finding the shortest path between two vertices (i.e., the shortest sequence of edge flips that transforms one drawing into another).  
The connectedness and the diameter have been studied for different families of planar straight-line graph drawings like  triangulations~\cite{hurtado1996flipping, lawson1972transforming, wagner2022connectivity},  spanning trees~\cite{aichholzer2024reconfiguration, avis1996reverse, bjerkevik2025flipping, bousquet2024reconfiguration, bousquet2023note, hernando1999geometric, nichols2020transition},  spanning paths~\cite{aichholzer2023flipping, AKL2007,CW2009,kleist2024connectivity}, and odd matchings~\cite{flips2024oddmatching}.
Regarding the complexity question,
it has been shown that finding the shortest flip sequence between triangulations is not only NP-hard~\cite{amp-fdtsp-15,lubiw2015flip}, but also APX-hard~\cite{PILZ2014589}. On the positive side, there exists an FPT algorithm which uses the length of the flip sequence as a parameter~\cite{kanj2017computing}. 

For simplicity, in what follows we use \emph{plane perfect matching} as a shorthand for planar straight-line drawing of a perfect matching. Note that, for plane perfect matchings, it is necessary to replace at least two edges at the same time to transform one matching into another. In this paper, we refer to this operation as a \emph{flip}, and use the according definition of flip graph.
It has been shown, if the point set is convex, then the flip graph of plane perfect matchings is connected and the shortest flip sequence between two given matchings is computable in polynomial time~\cite{hernando2001graphs}. On the other hand, for point sets in general position, 
whether the flip graph is connected or not is still an open question. 
Further research focuses on non-plane perfect matchings~\cite{biniaz2019flip, da2023longest} or other variants of flip operations~\cite{aichholzer2009compatible, aichholzer2022disjoint,milich2021flips}.

%In this paper, we study the problem of finding the shortest sequence of flips between two given perfect matchings on the same point set, called \flipmat:

In this paper, we study the problem of deciding whether a plane perfect matching can be transformed into another with at most $k$ flips, called \flipmat:

\medskip\noindent\fbox{%
  \parbox{0.98\textwidth}{
\flipmat\\
\textnormal{\textbf{Input:}} $\langle \ps, \ms,\mf,k \rangle$. A point set \ps, two plane perfect matchings $\ms$ and  $\mf$ on \ps, and a positive integer $k$.\\
\textnormal{\textbf{Question:}} Does there exist a sequence of at most $k$ flips which transforms $\ms$ into $\mf$? 
  }%
}

\medskip
%We remark that the number of points in \ps equals the number of vertices of $\ms$ and~$\mf$.
In~\cite{hernando2001graphs}, it is shown that \flipmat can be solved in linear time if \ps is convex, and the authors leave it as an open problem to study the case in which \ps is not convex. 

We solve this problem by proving that \flipmat is NP-hard, even under the additional restriction that the points of \ps have integer coordinates and the area of the minimum-size axis-aligned bounding box containing \ps is polynomial in the size of the matching. Our contribution is summarized by the following theorem.

\begin{theorem}\label{th:flipping-matchings-np-complete}
    \flipmat is NP-complete, even for integer point sets whose area is polynomial in the size of the matching.
\end{theorem}

Due to space limitations, we only describe the main ideas of the reduction; the complete proof is in the appendix.

\section{Proof of \cref{th:flipping-matchings-np-complete}}

Let $\m$ be a plane perfect matching on a point set \ps. Let $e_1$ and $e_2$ be two edges of $\m$ whose endpoints are mapped on a subset $\mathcal{Q}$ of \ps. A \emph{flip of $e_1$ and $e_2$} is an operation that eliminates $e_1$ and $e_2$ and introduces $e_3$ and $e_4$ such that their endpoints are still mapped on $\mathcal{Q}$, no edge crossing is introduced, and no two vertices are mapped to the same point. In other words, a flip operation produces a different plane perfect matching on \ps; see \cref{fig:flip} in which $\mathcal{Q}=\{p_1,p_2,q_1,q_2\}$. 
In what follows, we may use the term \emph{flip} without specifying the involved edges if they are clear from the context, or if we refer to a sequence of flips. 

\begin{figure}[htb]
 \centering
 \includegraphics[page=1]{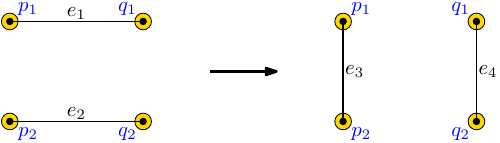}
 \caption{Flipping edges $e_1$ and $e_2$ to $e_3$ and $e_4$; points are drawn as yellow~circles.}
 \label{fig:flip}
\end{figure}

 % An \emph{integer point set} is a point set whose points have integer coordinates. The \emph{area} of an integer point set is the area of the minimum-size axis-aligned bounding box containing the point set.

The membership of \flipmat to NP is trivial, as one can guess all possible sequences of $k$ flips, and, for each of them, verify whether it transforms $\ms$ into~$\mf$. To prove hardness, we reduce from the NP-complete problem \planarvertexcover~\cite{DBLP:journals/siamam/GareyJ77}:

\medskip\noindent\fbox{%
  \parbox{0.98\textwidth}{
\planarvertexcover\\
\textnormal{\textbf{Input:}} $\langle  G=(V,E),c \rangle$. A planar graph $G=(V,E)$, and a positive integer $c$.\\
\textnormal{\textbf{Question:}} Does there exist a set of vertices $V_C \subseteq V$, called \emph{vertex cover} of $G$, such that $|V_C| \le c$ and every edge of $G$ has at least one vertex in $V_C$? 
  }%
}

    \begin{figure}[tbp]
      \centering
      \includegraphics[page=2]{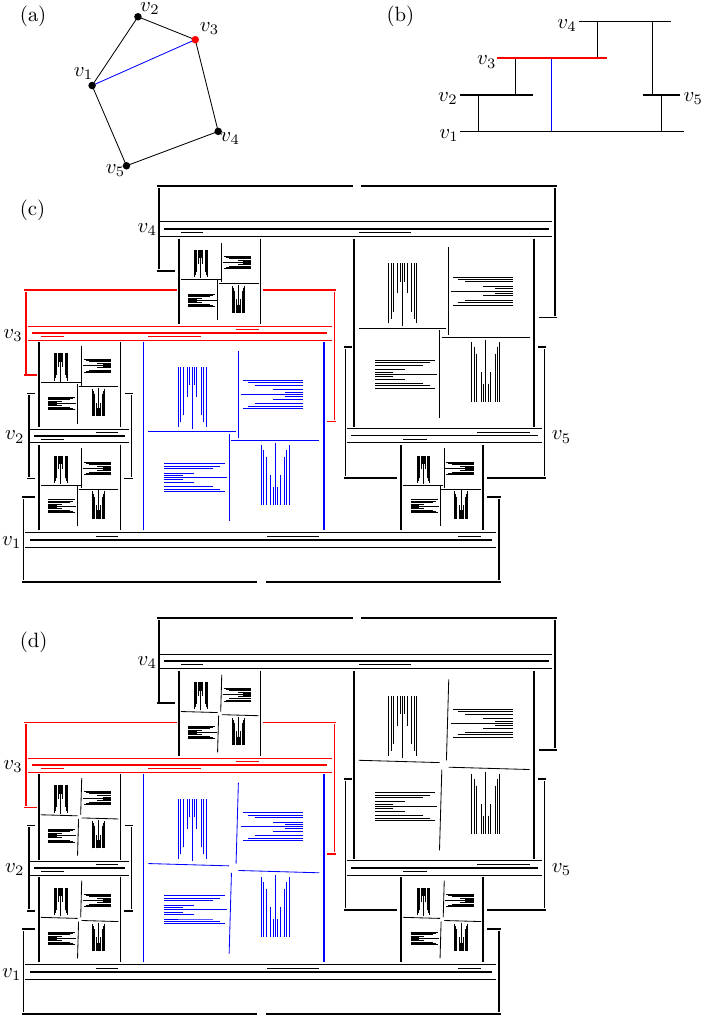}
    %   \begin{minipage}[t]{0.25\textwidth}
    %     \centering
    %     \includegraphics[page=1]{figures/example}
    %   \subcaption{$G$\label{fig:example-a}}
    % \end{minipage}
    %   \hfil
    %  \begin{minipage}[t]{.45\textwidth}
    %  \centering
    %  \includegraphics[page=2]{figures/example}
    %  \subcaption{$R$\label{fig:example-b}}
    %  \end{minipage}
    %  \begin{minipage}[t]{\textwidth}
    %  \centering
    %  \includegraphics[page=3,width=\textwidth]{figures/example}
    %   \subcaption{$\Gamma_1$\label{fig:example-c}}
    %  \end{minipage}
     \caption{(a) A graph $G$ where a vertex cover of size $3$ is depicted in gray. (b) A weak-visibility representation $R$ of $G$. (c) The plane perfect matching $\ms$ obtained by replacing each vertex-segment of $R$ by a vertex gadget and each edge-segment of $R$ by an edge gadget. (d) The plane perfect matching $\mf$. Each bold line represents $2k+2$ line segments. The vertex $v_3$ is red and the edge $v_1v_3$ is blue in each representation.}
      \label{fig:example}
     \end{figure}

\paragraph*{Construction.}

We begin by describing the construction behind our proof. Given an instance $I=\langle G=(V,E), c\rangle$ of \planarvertexcover, we construct an instance $I'=\langle \ps,\ms,\mf,k \rangle$ of \flipmat, with $k = 2c+5|E|$, as follows; refer to \Cref{fig:example}. 
At high level, in order to obtain $\ms$, the vertices and edges of $G$ are replaced by vertex gadgets and edge gadgets. Also, $\ms$ and $\mf$ only differ in the edge gadgets. Further, the fastest way to reconfigure an edge gadget should be through a constant-length flip sequence involving an edge of a vertex gadget. Our construction goes through a preliminary representation $R$ in which all vertices (edges) of $G$ are drawn as horizontal (vertical) segments.

\medskip\noindent\textbf{From $G$ to $R$.} Since $G$ is planar, it admits a weak-visibility representation~\cite{DBLP:journals/dcg/TamassiaT86}, namely, a representation of $G$ such that: each vertex is represented as a horizontal segment, each edge is represented as a vertical segment whose endpoints, called \emph{attachments}, lie on the horizontal segments of the corresponding endpoints, no two segments intersect each other except possibly at attachment points, all endpoints of horizontal and vertical segments are on an integer grid of quadratic size in the number of vertices. (The representation is called weak because two horizontal segments may see each other even when the corresponding vertices are not adjacent in the graph.)
\cref{fig:example}a shows a graph $G$ with a vertex cover of size three depicted in gray, and \cref{fig:example}b shows a weak-visibility representation $R$ of $G$.

\medskip\noindent\textbf{From $R$ to $\ms$.} Given a weak-visibility representation $R$ of $G$, to obtain a plane perfect matching $\ms$ and to define the corresponding point set~\ps, we replace each vertex-segment of~$R$ by a \emph{vertex gadget} and each edge-segment of $R$ by an \emph{edge gadget}; \cref{fig:example}c shows the drawing $\ms$ obtained from the weak-visibility representation~$R$ of \cref{fig:example}b. %Below we describe the vertex and edge gadget.

    An edge gadget consists of: $(i)$ The \emph{flip structure}, which contains four edges each having one endpoint in the central part of the gadget; $(ii)$ four \emph{blockers}, each consisting of eleven parallel edges; $(iii)$ two \emph{separators}, each consisting of $2k+2$ parallel edges. Refer to \cref{fig:configurations} where the edges of a separator are represented as bold blue line segments. 
    Intuitively, blockers prevent any two edges of the flip structure to be flipped within the same operation, while separators prevent any interplay between different edge gadgets.
    
   % \todo[inline]{The blockers are placed in such a way that no two edges of the flip structure can be flipped with each other. The separators are used such that it takes a large number of flips to  be able to flip edges of two different edge gadgets which are not separators.}
    
    %\begin{figure}[htb]
    % \centering
    % \includegraphics[page=1,width=.49\textwidth]{figures/edge-gadget}
    % \caption{An example of edge gadget. The edges of the flip structure are colored red, the edges of the blockers are colored black, and the edges of the separators are colored blue.}
    % \label{fig:edge-gadget}
    %\end{figure}

    \begin{figure}[tbp]
     \centering
     \begin{minipage}[t]{.45\textwidth}
     \centering
     \includegraphics[page=2,width=.8\textwidth]{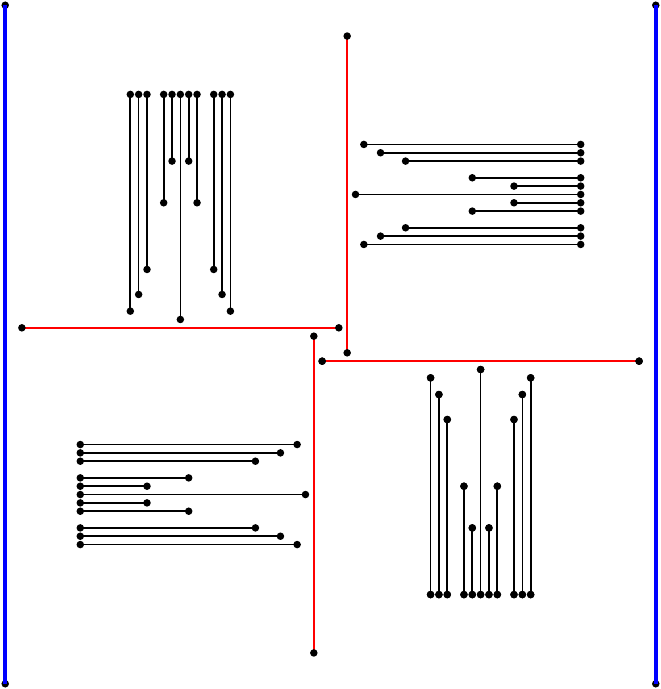}
    \subcaption{Edge gadget in the start-configuration.\label{fig:edge-gadget}}
     \end{minipage}\hfil
     \begin{minipage}[t]{.45\textwidth}
     \centering
    \includegraphics[page=3,width=.8\textwidth]{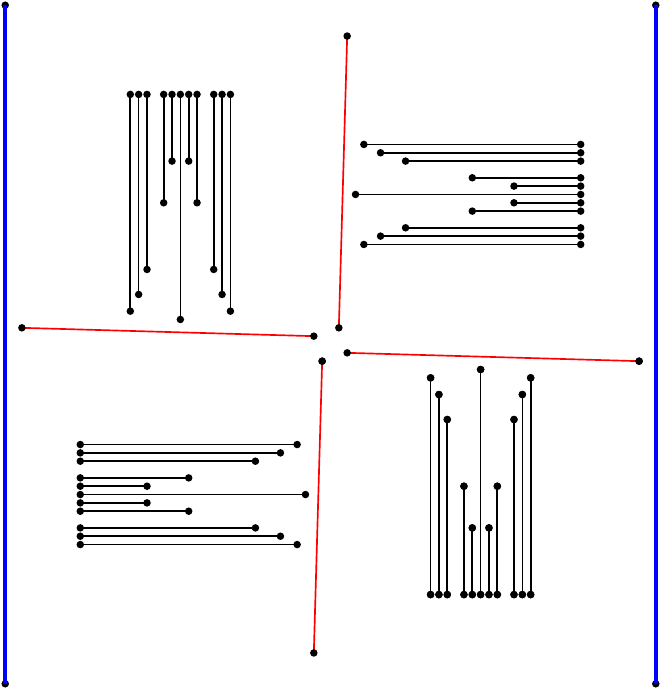}
    \subcaption{Edge gadget in the final-configuration.\label{fig:configurations-b}}
     \end{minipage}\hfil
     \caption{Construction of an edge gadget. The edges of the flip structure are colored red, the edges of the blockers are colored black, and the edges of the separators are represented as a blue line segment.}\label{fig:configurations}
    \end{figure}

A vertex gadget associated with a vertex-segment of $R$ representing a vertex of $G$ consists of three parts: One \emph{frame}, depicted in red in \cref{fig:deactivated}, and two \emph{vertex separators}, depicted in blue in \cref{fig:deactivated},  one for each of the two sides of the frame.
The frame consists of several horizontal line segments: $(i)$ The \emph{top-edge} and the \emph{bottom-edge}, which are the topmost edge and the bottommost edge, respectively, and they are the longest ones; $(ii)$ the $2k+2$ \emph{middle-edges}, which are shorter than both the top-edge and the bottom-edge; $(iii)$ the \emph{connectors}, which are edges that lie in the region between the top-edge and the middle-edges or in the region between the bottom-edge and the middle-edges. Each vertex separator consists of $2k+2$ vertical edges next to the frame and $2k+2$ horizontal edges above and $2k+2$ horizontal edges below these vertical edges. \cref{fig:activation} shows an example of a vertex gadget  incident to three edge gadgets. Intuitively, for every incident edge gadget, a connector is placed such that, after flipping the top-edge and the bottom-edge, the connector can be flipped with an edge of the flip structure of the edge gadget.
Vertex separators and middle-edges prevent any interplay between different vertex gadgets.

In what follows, we consider a particular flip operation, shown in \cref{fig:activation}, which transforms a vertex gadget from a \emph{deactivated} configuration (see \cref{fig:deactivated}) to an \emph{activated} configuration   (see \cref{fig:activated}). When performing this flip, we  also say that we \emph{activate} a vertex gadget and, conversely, we \emph{deactivate} a vertex gadget if we perform the reverse~operation.

     \begin{figure}[t]
      \centering
      \begin{minipage}{\textwidth}
      \centering
      \includegraphics[page=4, scale=0.65]{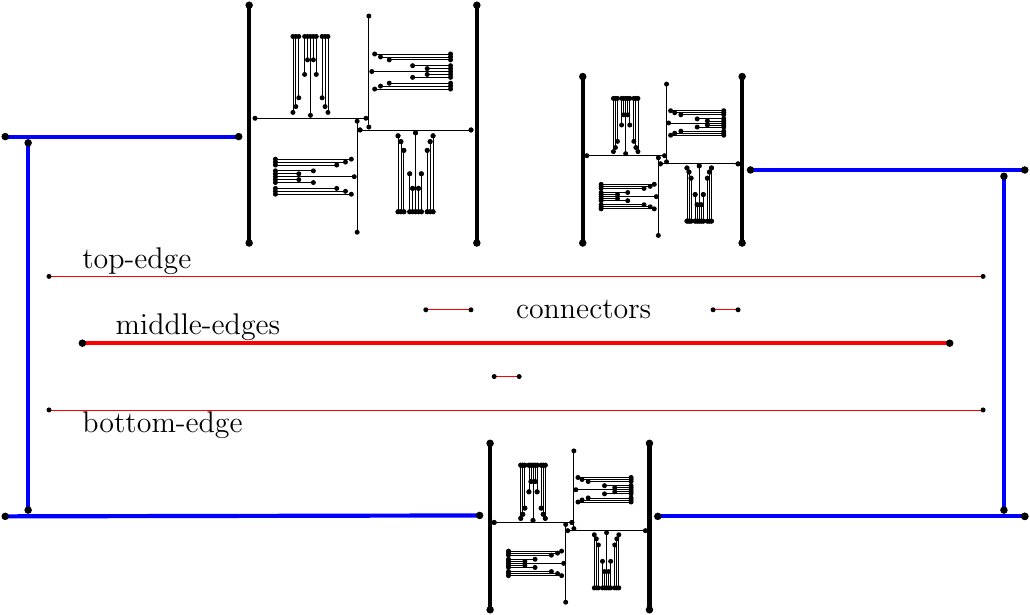}
      \subcaption{Deactivated vertex gadget. The bold edges are a set of $2k+2$ edges.\label{fig:deactivated}}
      \end{minipage}
      \begin{minipage}{\textwidth}
      \centering\includegraphics[page=5,scale=0.65]{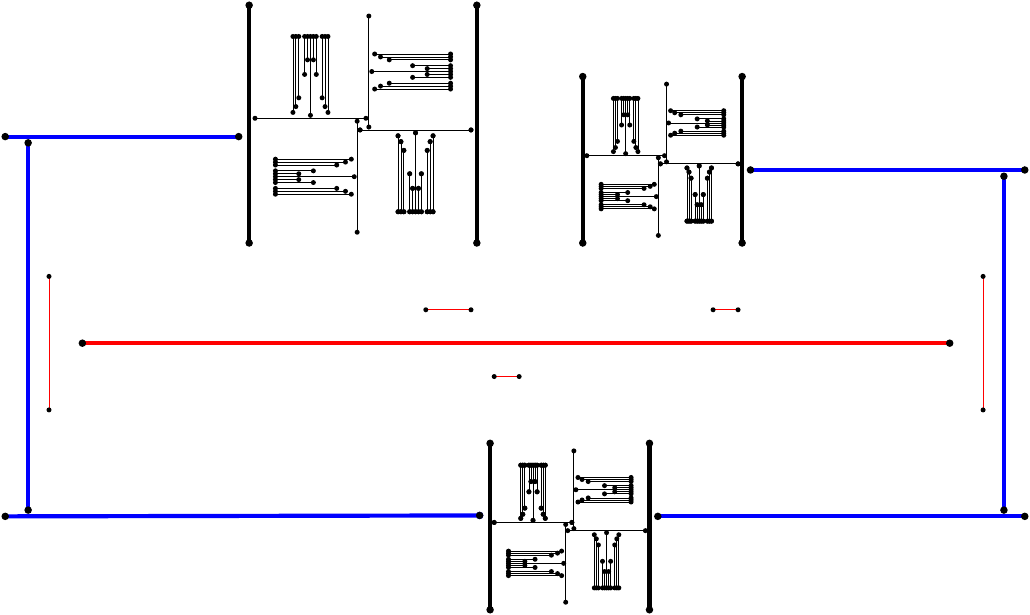}
      \subcaption{Activated vertex gadget.\label{fig:activated}}
      \end{minipage}
      \caption{Two notable configurations of a vertex gadget.}
      \label{fig:activation}
     \end{figure}
    
%    \begin{figure}[htb]
%     \centering
%     \includegraphics[page=1,width=.49\textwidth]{figures/edge-gadget}
%     \caption{An example of edge gadget. The edges of the flip structure are colored red, the edges of the blockers are colored black, and the edges of the separators are colored blue.}
%     \label{fig:edge-gadget}
%    \end{figure}

    Observe that to obtain $\ms$ we need to ``stretch'' the visibility representation~$R$, which corresponds to the introduction of  additional rows and columns in the underlying grid (which is always possible if the inserted row or column only crosses vertical or horizontal segments, respectively);
    refer to \cref{fig:example}c for an example.     
    More precisely, given two edges $(u,v)$ and $(u,w)$ of $G$ such that the edge-segment representing $(u,v)$ is longer than the edge-segment representing $(u,w)$ in $R$, the corresponding edge gadgets have different sizes in $\ms$ and, consequently, the vertex gadgets may need to have different lengths. 

\medskip\noindent\textbf{From $\ms$ to $\mf$.}    The plane perfect matching $\mf$ differs from $\ms$ only in the edge gadgets. More precisely, an edge gadget can assume two configurations, which we call the \emph{start-configuration} and the \emph{final-configuration}, based on the mapping of the four edges of the flip structure. \cref{fig:edge-gadget} and \cref{fig:configurations-b} show the start-configuration and the final-configuration of an edge gadget, respectively. In~$\ms$ all the edge gadgets are in the start-configuration, whereas in $\mf$ all the edge gadgets are in the final-configuration. A start-configuration can be transformed into a final-configuration by a sequence of five flips. Such a sequence of flips transforms an edge gadget associated with an edge $e$ of $G$ by using a connector of a vertex gadget that is associated with one of the endpoints of $e$. Observe that it is possible to use a connector of a vertex gadget only if the vertex gadget is activated. \cref{fig:5flips} shows the sequence of five flips to transform an edge gadget from its start-configuration to its final-configuration. 

 \begin{figure}[tbp]
  \centering
  \includegraphics[page=1,scale=0.9]{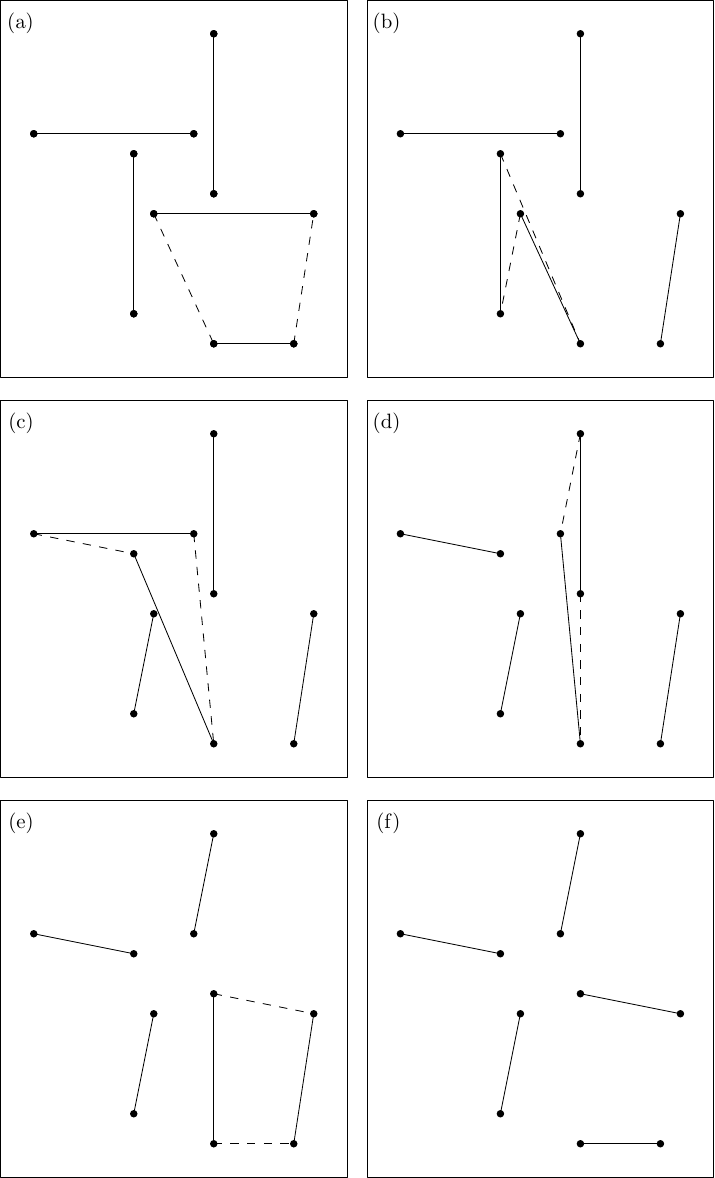}
       \caption{(a)--(e) A sequence of 5 flips to transform an edge gadget from the start-configuration to the final-configuration.  (f) The final-configuration. In all subfigures, dashed edges are the ones that are added with the flip operation, and the bottommost edge in (a) and (f) is a connector of the vertex gadget. For the sake of readability, we only illustrate the flip structure of the edge gadget and the connector of the vertex gadget that is used to perform the transformation.} \label{fig:5flips}
 \end{figure}

    \medskip\noindent\textbf{The point set \ps.} The described point set is of polynomial size, namely $|\ps| \in O(|V|^2+|E|^2)$. Also, it is an integer point set, because it is constructed starting from a weak-visibility representation on an integer grid by introducing additional rows and columns when needed. However, the exact placement of the points depend on the length of the segments representing the edges of the matching, which are discussed in the appendix and yield polynomial~area.

\paragraph*{Correctness.}

We now show the correctness of our reduction in the following two lemmas.

%We now show that $I'=\langle \ps,\ms,\mf,2c+5|E| \rangle$ is a \texttt{yes} instance of \flipmat if and only if $I=\langle G=(V,E), c\rangle$ is a \texttt{yes} instance of \planarvertexcover.

\begin{lemma}\label{le:first-direction}
     If $I=\langle G=(V,E), c\rangle$ is a \texttt{yes} instance of \planarvertexcover then $I'=\langle \ps,\ms,\mf,2c+5|E| \rangle$ is a \texttt{yes} instance of \flipmat.
\end{lemma}
\begin{proof}
Let $V_C \subseteq V$ be a vertex cover of $G$ of size $c' \le c$. We consider the $c'$ vertex-segments of $R$ corresponding to the vertices of $V_C$ and we activate each of the corresponding vertex gadgets of $\ms$. For each edge-segment of $R$ incident to these vertex-segments, we transform the corresponding edge gadget of $\ms$ from the start-configuration to the final-configuration, as shown before. Once this has been done for each edge gadget, we deactivate the vertex gadgets that we previously activated, which yields to $\mf$. The  process requires exactly $2c'+5|E|$ flips in total. Indeed, $(i)$ for each edge gadget of $\ms$, exactly 5 flip operations are required to transform the start-configuration into the final-configuration; $(ii)$ to do these transformations, exactly $c'$ vertex gadgets need to be activated and deactivated. 
\end{proof}

\begin{lemma}\label{lem:flip_to_vc}
 If $I'=\langle S,\ms,\mf,2c+5|E| \rangle$ is a \texttt{yes} instance of \flipmat then $I=\langle G=(V,E), c\rangle$ is a \texttt{yes} instance of \planarvertexcover.    
\end{lemma}
\begin{proof}[Proof sketch]
     We argue that the only way to transform $\ms$ into $\mf$ with a sequence of $2c+5|E|$ flips is through the activation of at most $c$ vertex gadgets such that all edges of $G$ have at least one vertex among the corresponding vertices. In other words, a feasible sequence of flips must identify at most $c$ vertices forming a vertex cover of $G$. 
     Due to the presence of the blocking structures and the separators for the edge gadgets and the vertex separators for the vertex gadget, we cannot transform $\ms$ to $\mf$ with a sequence of less than $2c+5|E|$ flips.
     The proof is rather technical and completely deferred to the appendix.
\end{proof}

\section{Conclusion and Further Research}
We have shown that deciding whether $k$ flips suffice to transform a given plane perfect matching into another on the same point set is NP-complete. 
While the point set exploited in our reduction has integer coordinates and occupies polynomial area, it is not in general position. Therefore, a natural question is whether \cref{th:flipping-matchings-np-complete} holds for point sets in general position. %In addition, it would be interesting to study flip graphs also in non-general position. We conjecture that most results for general position also hold if we allow non-general position. 

Recall that the connectedness of the flip graph of perfect matchings is still open. Toward solving this question, we are currently working on the setting in which our point set consists of at most two convex layers. 
Moreover, we are working on extending our NP-hardness proof for plane odd matchings and plane spanning paths.

%Another open problem is whether the flip graph of matchings is always connected. It is also not proven whether a flip always exists.

\bibliography{biblio}

%-----------------------------------------------------------------
% START OF APPENDIX
%-----------------------------------------------------------------

\clearpage

\appendix

\section{Proof of Lemma~\ref{lem:flip_to_vc}}

As already mentioned, we argue that the only way to transform $\ms$ into $\mf$ with a sequence of $2c+5|E|$ flips is by activating a set of at most $c$  vertex gadgets associated with vertices forming a vertex cover of $G$. In the proof we use some lemmata whose proofs can be found in the next sections, as they require a more detailed description of the edge and vertex gadgets.

Assume that we have a sequence of flips that transforms $\ms$ to $\mf$. 
By Lemma~\ref{lem:2k_crossings},  we need at least $k+1$ flips (hence exceeding our budget) if we have an edge between a point of the flip structure or the blockers of an edge gadget and a point that is not in the edge gadget or an adjacent vertex gadget. Consequently, all flips used to transform a start-configuration  into a final-configuration of an edge gadget only use edges connecting two points of the edge gadget itself or an adjacent vertex gadget.
Note that this means that the flip structures of two different edge gadgets are flipped independently (i.e., the corresponding flip sequences do not share operations).

By Lemma~\ref{lem:flips_blockers} and Lemma~\ref{lem:separator}, we need at least eight flips for each edge gadget to transform its start-configuration $\mathcal{F}_s$ into its final-configuration $\mathcal{F}_f$ only using points of the edge gadget itself. Note that, since $c \leq |E|$, it holds $8|E|>2c+5|E|=k$.

Let $\mathcal{V}$ be a vertex gadget.
By construction, no edge of $\mathcal{V}$ sees an edge of the blockers $\mathcal{B}$, the start-configuration $\mathcal{F}_s$, and the final-configuration $\mathcal{F}_f$ of any edge gadget.
Hence, we need at least two flips in every used vertex gadget, one flip to obtain an edge that sees an edge of $\mathcal{F}_s$, and one flip to get rid of the edge that sees an edge of $\mathcal{F}_f$. Note that these flips are independent for each used vertex gadget, but different adjacent edge gadgets might use the same flip. Hence, we have at least two flips in each used vertex gadget, independent of the number of edge gadgets incident to it.
For each edge gadget where points of $\mathcal{V}$ are involved in flipping the flip structure, we need at least five flips if only one edge is involved, or at least six flips if at least two other edges are involved.

Let $\mathcal{E}_i$ be the set of edge gadgets where $\mathcal{F}_s$ is transformed into $\mathcal{F}_f$ only using edges of the corresponding edge gadget.
Let $\mathcal{E}_o$ be the set of edge gadgets where $\mathcal{F}_s$ is transformed into $\mathcal{F}_f$ by using edges outside of the corresponding edge gadget, and let $\mathcal{V}_o$ be the set of vertex gadgets whose edges are used by edge gadgets in $\mathcal{E}_o$.
Note that $|\mathcal{E}_i|+|\mathcal{E}_o|=|\mathcal{E}|$ $(= |E|)$.
By the observations above, we need at least $8|\mathcal{E}_i|+5|\mathcal{E}_o|+2|\mathcal{V}_o|$ flips. Note that $|\mathcal{E}_i|$ is at least the size of a vertex cover of the corresponding induced subgraph in $G$. Hence $3|\mathcal{E}_i| + 2|\mathcal{V}_o|>c$ if $\mathcal{E}_i$ is not empty.
Therefore, $\mathcal{E}_i$ is empty and $\mathcal{E}_o$ is the set of all edge gadgets.
Then every edge gadget is adjacent to some vertex gadget in $\mathcal{V}_o$.
So $\mathcal{V}_o$ corresponds to a vertex cover of size at most $c$ in $G$.

\section{Technical details}

In this section we describe the constructions of the gadgets in more detail and prove that the gadgets have the desired properties exploited in the previous overview.
Before we describe the gadgets, we state two lemmata that are used to show properties of the gadgets.

\begin{lemma}\label{lem:crossings}
 Let \ps be a point set and let \m and $\mathcal{M}'$ be plane perfect matchings on \ps.
 Let $e$ be an edge of $\mathcal{M}'$ and let $x$ be the number of edges of $M$ crossing $e$.
 Then we need at least $\frac{x}{2}$ flips to transform \m into $\mathcal{M}'$.
\end{lemma}

\begin{proof}
There are $x$ edges in \m that cross $e$.
Hence, if an intermediate matching crosses $e$, then it is not $\mathcal{M}'$.
In every step, we flip two edges.
So we can flip at most two edges that cross $e$ in every step.
This means that we reduce the number of edges that cross $e$ by at most two in every step.
Therefore we need at least $\frac{x}{2}$ flips to transform \m into~$\mathcal{M}'$.
\end{proof}

\begin{lemma}\label{lem:flip_4gon}
Let \m be a plane perfect matching, and let $a$ and $c$ be two edges of \m.
If we flip $a$ and $c$ and obtain an edge $b$, then the other edge $d$ is uniquely defined.
Furthermore, $abcd$ is a plane $4$-gon.
\end{lemma}

\begin{proof}
    Let $a=v_1v_2$ and $c=v_3v_4$.
    Without loss of generality, let $b=v_2v_3$.
    Let $\mathcal{M'}$ be a perfect matching that can be obtained by flipping $a$ and $c$ in \m and that contains the edge $b$.
    Since $\mathcal{M'}$ is again a perfect matching, $d$ is $v_1v_4$ since $v_1$ and $v_4$ are the only points who are not matched after deleting $a$ and $c$ from \m and adding $b$.

    Regarding the plane $4$-gon, $a$ and $c$ do not cross $b$ or $d$ since they are adjacent.
    Since \m and $\mathcal{M'}$ are plane, $a$ does not cross $c$ and $b$ does not cross $d$.
\end{proof}

In what follows, we say that a point $p$ \emph{sees} a point $q$ if the segment $\overline{pq}$ does not cross any edge of the matching.
Further, we say that an edge $v_1v_2$ \emph{sees} an edge $v_3v_4$ if $v_1$ sees $v_3$ and $v_2$ sees $v_4$, or $v_1$ sees $v_4$ and $v_2$ sees $v_3$.
We say that an edge is \emph{involved} in a flip sequence if it is flipped at least once.

We next analyze the edge gadgets in more detail. As stated before, each edge gadget consists of the flip structure, four blockers, and two separators.
\begin{itemize}
    \item $F$ denotes the set of vertices spanning the flip structure.
    \item  $\mathcal{F}_s$ and $\mathcal{F}_f$ denote the start-configuration and the final-configuration, respectively.
    \item The four blockers are denoted by $\mathcal{B}_0, \mathcal{B}_1, \mathcal{B}_2$, and $\mathcal{B}_3$, and we call their union $\mathcal{B}$ the \emph{blocking structure}
    \item The left and right separators are denoted by $\mathcal{S}_l$ and $\mathcal{S}_r$, respectively
\end{itemize}

\subsection{The flip structure}
The flip structure $F$ consists of eight points, four inner and four outer ones.
Let $f_0, f_1, f_2, f_3$ be the outer points in clockwise direction and $f'_0, f'_1, f'_2, f'_3$ be the inner points, such that $f_if'_i$ is an edge in $\mathcal{F}_s$, for $0\leq i \leq 3$ and $f_0$ is the topmost of $F$. 
In $\mathcal{F}_s$, $f_i$ sees $f'_{i-1}$, but does not see $f'_{i+1}$ and $f'_{i+2}$ for $0\leq i \leq 3$.
Further, $\mathcal{F}_f$ consists of the edges $f_if'_{i-1}$ for $0\leq i\leq 3$. The indices for the points are modulo $4$. 

In the following lemma we prove lower bounds for the number of flips to get from $\mathcal{F}_s$ to~$\mathcal{F}_f$.

\begin{lemma}\label{lemma:flipping_flip_structure}
Let $\mathcal{M}=\mathcal{F}_s \cup \mathcal{R}$ be a plane perfect matching that consists of $\mathcal{F}_s$ and some other edges $\mathcal{R}$.
\begin{itemize}
\item It takes at least four steps to flip $\mathcal{F}_s$ to $\mathcal{F}_f$ if  no edge of $\mathcal{R}$ is involved in any of the flips.
 \item It takes at least five steps to flip $\mathcal{F}_s$ to $\mathcal{F}_f$ if exactly one edge of $\mathcal{R}$ is involved in a flip.
 \item It takes at least six steps to flip $\mathcal{F}_s$ to $\mathcal{F}_f$ if at least two edges of $\mathcal{R}$ are involved in the flips, but no flip involves only edges of $\mathcal{R}$.
\end{itemize}
\end{lemma}

\begin{proof}

 Let $S'$ be a point set that contains $F$ such that $\mathcal{F}_s$ and $\mathcal{F}_f$ can be extended to plane perfect matchings.
 Let \m be a perfect matching on $S'$.
 We define the \emph{weight} of \m as the number of edges of $\mathcal{F}_f$ in \m minus the number of  edges of $\mathcal{F}_s$ in \m. Since $\mathcal{F}_s$ and $\mathcal{F}_f$ each consists of exactly four edges, the weight of \m is between $-4$ and $4$.
 
$\mathcal{F}_s$ has weight $-4$ and $\mathcal{F}_f$ has weight $4$.
So we need to increase the weight of the configuration by $8$. In order to analyze the increment of weight of a single flip operation, we next discuss three cases depending on which edges we flip and which we obtain after each flip. 

 \begin{itemize}
  \item \textsc{Case 1:} We flip two edges $a$ and $c$ of $\mathcal{F}_s$.

  Then we do not obtain any edge of $\mathcal{F}_f$.
  Assume to the contrary that we obtain an edge $b$ of $\mathcal{F}_f$.
  By \Cref{lem:flip_4gon}, $a$, $b$, and $c$ define a plane 4-gon.
  Any of such 4-gons
  %4-gon spanned by two edges of $\mathcal{F}_s$ and one edge of $\mathcal{F}_f$
  contains only one point by construction. This point is matched with another point, that is outside of the $4$-gon and this edge crosses the $4$-gon.
  Hence, there is a matching in the flip sequence that is not plane. Therefore, we never obtain an edge of $\mathcal{F}_f$ if we flip two edges of $\mathcal{F}_s$.

   \item \textsc{Case 2:} We flip an edge of $\mathcal{F}_s$ and any other edge.

  By construction, any $4$-gon
  defined
  %spanned 
  by an edge of $\mathcal{F}_s$, two edges of $\mathcal{F}_f$ and another edge is self crossing.
  Hence, we obtain at most one edge of $\mathcal{F}_f$.

   \item \textsc{Case 3:} We flip two edges not in $\mathcal{F}_s$.

  Since these two edges contain four points, we obtain at most two edges of $\mathcal{F}_f$.
 \end{itemize}

Based on the above cases, each flip operation increases the weight by at most $2$.
 Hence, we need at least $4$ flips to get from $\mathcal{F}_s$ to $\mathcal{F}_f$, which proves the first item of the statement.

Now consider the case that at least one edge that does not connect points of the flip structure is involved in the shortest flip sequence.
Let $e$ be an edge connecting two points that are not in the flip structure.
If we flip $e$ with any edge in the flip structure, we do not get an edge of $\mathcal{F}_f$.
Hence, we only increase the weight of the configuration by at most one.
This means, we need at least $5$ flips to get from $\mathcal{F}_s$ to $\mathcal{F}_f$ if  another edge is involved in the shortest flip sequence, which proves the second item of the statement.

Now let $e'$ be another edge connecting two points not in $F$ which is involved in the shortest flip sequence.
Then, as with $e$, any flip involving $e'$ increases the weight of the current configuration by at most one. Also, note that then at least one endpoint of $e$ or $e'$ is matched with a point of the flip structure. Since every point of $F$ is matched with another point of $F$ in $\mathcal{F}_f$, every edge with one point in $F$ and one point in $e$ or $e'$ has to be flipped.
Consider the last flip that involves such an edge.
In this flip we have two edges each having exactly one point in $F$. Otherwise, we still have an edge with exactly one point in $F$.
Therefore, we do not remove any starting edge and get at most one final edge.
So we increase the weight of the configuration by at most one again.
This means, we have at least three flips which increase the weight of the configuration by at most one.
Since the remaining flips increase the weight of the configuration by at most two, we need at least six flips in this case ($3\cdot 1 + 2\cdot 2<8$), which shows the third item of the statement and concludes the proof.
\end{proof}

\subsection{The blockers}
The main goal for the blockers is to prevent that in an edge gadget, $\mathcal{F}_s$ can be flipped to $\mathcal{F}_f$ in less than eight flips.

To achieve this, we need edges that cross $f_if_{i+1}$, otherwise we can transform $\mathcal{F}_s$ into $\mathcal{F}_f$ in four steps.
Also, we want to ensure that we need multiple flips to obtain an edge that sees $f_if'_i$. Furthermore, we place the points of a blocker in such a way that each of these points sees at most one of the inner points of $\mathcal{F}_s$.
We now describe the construction for $\mathcal{B}_1$; the other parts of the blocking structure are rotated copies of $\mathcal{B}_1$. Refer to \cref{fig:blocker_structure} for an illustration.
% Hence we use the following construction for each $\mathcal{B}_i$.
%\cref{fig:blocker_structure} depicts the construction of a blocker; the horizontal edge is $f'_1f_1$.
%In the following we describe the construction of $\mathcal{B}_1$, since all other $\mathcal{B}_i$ are just rotated copies of it.

\begin{figure}
 \centering
 \includegraphics[page=1]{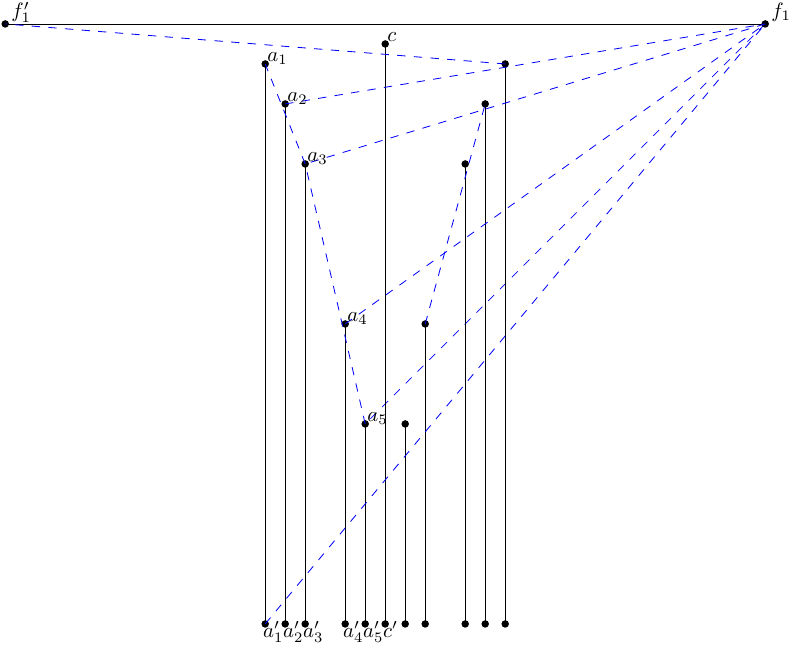}
 \caption{The drawing of a blocker. The line segment $f_1'f_1$ is part of $\mathcal{F}_s$. The blue dashed lines indicate the structure of the point set. The bottommost points are depicted on a line for presentation purposes.}
 \label{fig:blocker_structure}
\end{figure}

Note that the construction is symmetric.
We label the points on the right side of $cc'$ such that $b_j$ has the same height as $a_j$, and we have edges $b_jb'_j$.
Although the points $a'_j, c'$ and $b'_j$ are depicted as collinear, they form a convex point set.
Further, the points $a_j$ and $a'_j$ form a convex point set, and hence, also the points $b_j$ and $b'_j$ form one.
Each line $a'_jf_1$ crosses every edge $b_lb'_l$, $1\leq l \leq 5$, and $cc'$.
Each line $a_jf_1$ crosses every edge $b_lb'_l$ with $l<j$ and $cc'$.
Furthermore, $f_if_{i+1}$ crosses $a_ja'_j, cc', b_jb'_j$ for $3\leq j \leq 5$.
Hence, $f_if_{i+1}$ crosses seven edges.

In the following lemma, we show that we do not obtain an edge connecting of $\mathcal{B}$ that sees $f_if'_i$ in one flip.

\begin{lemma}\label{lem:flips_blockers}
Given an edge gadget $\mathcal{E}$ and its blocking structure $\mathcal{B}$.%\todo{Changed to be more similar to Lemma 8. Please check} %It takes at least two flips in $\mathcal{B}$ such that an edge connecting two points of $\mathcal{B}$ sees an edge of $\mathcal{F}_s$.
Let $\mathcal{E}_f$ be $\mathcal{E}$ with $\mathcal{F}_f$ instead of $\mathcal{F}_s$.
Then any sequence of flips that transforms $\mathcal{E}$ into $\mathcal{E}_f$, that only flips edges connecting points of $\mathcal{B}$ and $\mathcal{F}_s$, has length at least $8$.
\end{lemma}

\begin{proof}
We first show that we need at least $8$ flips to transform $\mathcal{F}_s$ to $\mathcal{F}_f$ if $f_if'_i$ sees $f_{i+1}f'_{i+1}$.
By construction, $f_if_{i+1}$ crosses seven edges.
Hence, we need at least $4$ flips that $f_if'_i$ sees $f_{i+1}f'_{i+1}$ by Lemma~\ref{lem:crossings}.
Further, we need at least $4$ flips to transform $\mathcal{F}_s$ into $\mathcal{F}_f$ by Lemma~\ref{lemma:flipping_flip_structure}. In total we need at least $8$ flips in this case.

We now show that we need at least two flips to obtain an edge that sees an edge of $\mathcal{F}_s$.
 By construction, the edge of $\mathcal{F}_s$ depicted in Figure~\ref{fig:blocker_structure} is the only edge of $\mathcal{F}_s$ that can be seen by any point of the blockers.
 Further, the edge $cc'$ prevents that any edge connecting two points on its left sees $f_if'_i$.
 Hence, if after some flips an edge connecting two points of $\mathcal{B}$ sees $f_if'_i$, then either this edge is $cc'$ or $cc'$ has been flipped.

 \begin{itemize}
    \item \textsc{Case 1:} $cc'$ sees $f_if'_i$.

  For this to happen, we need at least three flips by Lemma~\ref{lem:crossings}.
%\todo{Add lemma for crossings: DONE!}
   \item \textsc{Case 2:} $cc'$ has been flipped before.

  If $cc'$ is flipped with an edge of $\mathcal{B}$, then we can only flip it with the edge $a_5a'_5$ or $b_5b'_5$. Hence, we do not get an edge that sees $f_if'_i$, so we need at least two flips.
  If $cc'$ is involved in the second flip or later, then we have at least two flips already.
 \end{itemize}

Thus, in both cases we need at least two flips to obtain an edge that sees $f_if'_i$. Observe that this intermediate matching differs from $\mathcal{B}$ in at least two edges since two perfect matchings always differ in at least two edges. By Lemma~\ref{lemma:flipping_flip_structure} we need at least $6$ more flips to obtain $\mathcal{F}_f$ if two edges not connecting points of $F$ are involved in transforming $\mathcal{F}_s$ to $\mathcal{F}_f$.
Otherwise, if one edge not connecting points of $F$ is involved in transforming $\mathcal{F}_s$ to $\mathcal{F}_f$, we need at least $5$ more flips to obtain $\mathcal{F}_f$ . Since we still have an edge that is not in $\mathcal{B}$, we need at least one more flip to obtain $\mathcal{E}_f$.
In any case we need at least $8$ flips to transform $\mathcal{E}$ to $\mathcal{E}_f$.
\end{proof}

\subsection{The separators}

An edge of $\mathcal{F}_s$ might be seen by an edge of another edge gadget even without doing any flips.
To avoid this, we introduce separators. Each separator consists of $2k+2$ vertical edges whose topmost points are above $f_0$ and the bottommost points are below $f_2$. (Recall that $k=2c+5|E|$.)
%that are longer than the height of $F$. 
The two innermost edges are two units (one per side) longer than the other edges of a separator.
We denote the separator with $\mathcal{S}$, consisting of the left part $\mathcal{S}_l$ and the right part $\mathcal{S}_{r}$.
Furthermore, the separators are placed in such a way that for each edge $v_1v_2$ of $\mathcal{S}$ each of the lines $v_1f'_i$ and $v_2f'_i$ ($0\leq i \leq 3$) crosses at least seven edges of $\mathcal{B}$.
Note that by construction any edge of $\mathcal{S}$ sees only other edges of $\mathcal{S}$, except for the leftmost edge in $\mathcal{S}_l$ and the rightmost edge in $\mathcal{S}_r$.
In the following lemma we prove that, by adding $\mathcal{S}$ to edge gadgets, we need at least eight flip operations if we transform $\mathcal{F}_s$ into $\mathcal{F}_f$ by using edges of $\mathcal{S}$.

\begin{lemma}\label{lem:separator}
Let $\mathcal{E}$ be an edge gadget, let $\mathcal{F}_s$ and $\mathcal{F}_f$ be the start configuration and the final configuration of its flip structure, respectively, and let $\mathcal{S}$ be the separator of $\mathcal{E}$.
Then any sequence of flips that transforms $\mathcal{F}_s$ into $\mathcal{F}_f$ in $\mathcal{E}$ and involves an edge of $\mathcal{S}$ has length at least $8$.
\end{lemma}

\begin{proof}
By Lemma~\ref{lemma:flipping_flip_structure}, we know that we need at least $5$ flips.
Let $\mathcal{E}'$ be the edge gadget after some flips such that $\mathcal{F}_s \subset \mathcal{E}'$ and an edge of $\mathcal{E}' \setminus \mathcal{F}_s$ sees an edge of $\mathcal{F}_s$.

We distinguish two cases, based on whether the edge seeing $\mathcal{F}_s$ is an edge of $\mathcal{S}$ or not.

\begin{itemize}
     \item \textsc{Case 1:} An edge $v_1v_2$, where $v_1$ and $v_2$ are points of $\mathcal{S}$, sees an edge of $\mathcal{F}_s$.

    By construction, each of $v_1f'_i$ and $v_2f'_i$ is crossed at least seven times in $\mathcal{E}$.
    Hence, we need four flips to get from $\mathcal{E}$ to $\mathcal{E}'$ by \cref{lem:crossings}.
    We also need at least five flips to transform $\mathcal{F}_s$ into $\mathcal{F}_f$ in $\mathcal{E}'$ by \cref{lemma:flipping_flip_structure}.
    Hence, we need at least nine flips in this case.

     \item \textsc{Case 2:} An edge $v_1v_2$ with $v_1$ in $\mathcal{S}$ and $v_2$ in $\mathcal{B}$ sees an edge of $\mathcal{F}_s$.

Observe that any edge of $\mathcal{S}$ only sees edges that belong to $\mathcal{S}$.
This also holds for the leftmost edge $e_ae_b$ in $\mathcal{S}_r$ since the topmost point of $e_ae_b$ only sees $\mathcal{B}_0$, $f_0$ and $f_1$ while the bottommost point of $e_ae_b$ sees $\mathcal{B}_1$, $f_1$ and $f_2$. Similarly, the rightmost edge in $\mathcal{S}_l$ only sees an edges of $\mathcal{S}_l$.
Hence, in the first flip that involves an edge of $\mathcal{S}$ we obtain only edges that connect points of $\mathcal{S}$.
We need a second flip to obtain an edge with one point in $\mathcal{S}$ and one point in $\mathcal{B}$.
Also, observe that we now have two edges that are not in $\mathcal{E}$.
We need at least five flips to transform $\mathcal{F}_s$ into $\mathcal{F}_f$ in $\mathcal{E}'$ by \cref{lemma:flipping_flip_structure}.
If we only needed five flips, then we still have to do at least one more flip to obtain $\mathcal{E}$ with $\mathcal{F}_s$.
Hence, we need at least eight flips in this case. 
\end{itemize}
\end{proof}

\subsection{The vertex gadget}
We now turn our attention on the vertex gadget,  which consists of two parts.
\begin{itemize}
    \item The \emph{vertex frame} $\mathcal{H}$ consists of the bottom-edge, $2k+2$ middle-edges, the top-edge and several connectors. For every incident edge gadget, we have a connector; the precise placement is explained later.
    \item On the left and right side of $\mathcal{H}$ we have the \emph{vertex separators}, denoted as $\mathcal{I}_l$ and $\mathcal{I}_r$, respectively. Each vertex separator consists of $2k+2$ vertical edges, $2k+2$ horizontal edges above them, and $2k+2$ horizontal edges below them.
\end{itemize}
We describe both parts in the following with more details.

\begin{figure}
    \centering
    \includegraphics[page=3]{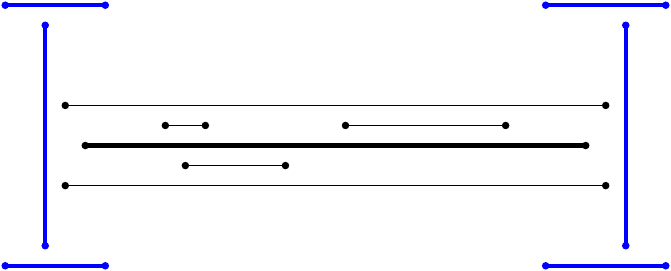}
    \caption{The vertex gadget. The frame $\mathcal{H}$ is drawn in black and the vertex separators $\mathcal{I}$ are drawn in blue. The heavier edges are each a set of $2k+2$ edges.}
    \label{fig:vertex_gadget2}
\end{figure}

\subsubsection{The vertex frame $\mathcal{H}$}
The vertex frame $\mathcal{H}$ consists of the top-edge, the bottom-edge, $2k+2$ middle-edges and several connectors.
The top- and the bottom-edge have the same lenght, while the middle-edges are two units (one per side) shorter than the top-edge. The leftmost (rightmost) points of each of these lines are to the left (right) of the leftmost (rightmost) incident edge gadget, respectively.
The middle-edges are such that the top- and the bottom-edge see each other.
Next, we describe the placement of the connectors, which are the edges that can be used to transform the flip structure in $5$ steps. The connectors are between the top-edge and the hightest middle-edge, or between the lowest middle-edge and the bottom-edge.

\begin{figure}
    \centering
    \includegraphics{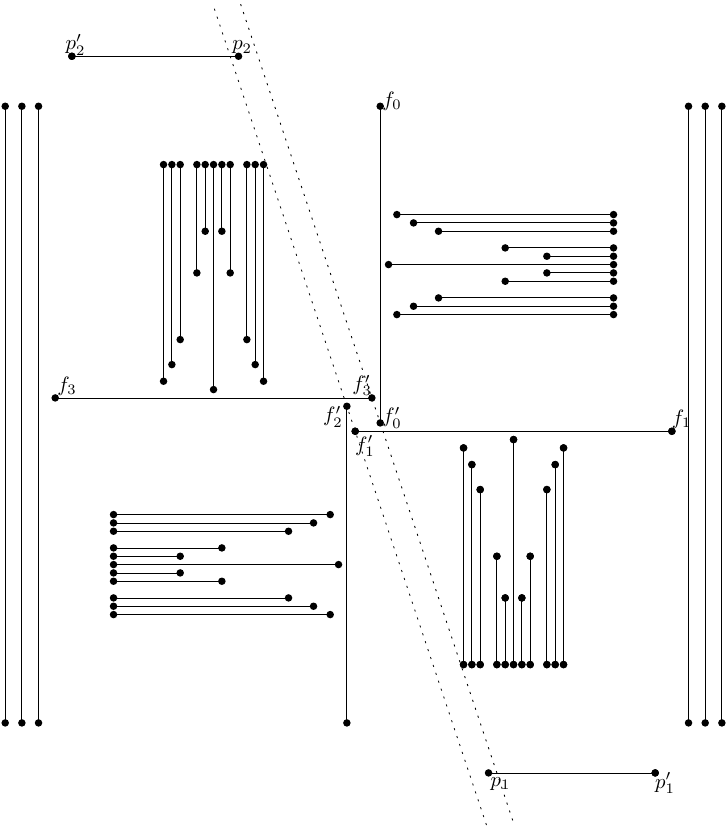}
    \caption{Placement of the connectors. Only the connectors are drawn of each vertex gadget.}
    \label{fig:connectors}
\end{figure}

For the precise placement of the connectors, consider the edge gadget $\mathcal{E}$ corresponding to an edge $e=(v_a,v_b)$ of $G$ and the vertex gadgets $\mathcal{V}_a$ and $\mathcal{V}_b$ corresponding to $v_a$ and $v_b$.
Without loss of generality, let $\mathcal{V}_a$ be above $\mathcal{E}$ and $\mathcal{V}_b$ below it.
We construct the connectors of $\mathcal{V}_a$ and $\mathcal{V}_b$.
See \cref{fig:connectors} for a picture of the edge gadget and the connectors.

The lines  $f'_0f'_3$ and $f'_1f'_2$ are parallel.
Note that no edge of $\mathcal{B}$ crosses any of these two lines.
For the connector in $\mathcal{V}_b$, we place a point $p_1$ between the lines  $f'_0f'_3$ and $f'_1f'_2$ and between the middle- and top-edge. For the second point of the connector, we place a point $p'_1$ on the same height as $p_1$ and below $f_1$.
We construct the connector in $\mathcal{V}_a$ similarly, but place the connector $p_2p'_2$ between the bottom- and middle-edge, and place the second point $p'_2$ above $f_3$.

Now $p_1$ has the property that the line segments $p_1f'_1$, $p_1f'_2$, $p_1f'_3$, and $p_1f'_0$ are in clockwise order around $p_1$. Hence, $p_1$ sees all points $f'_i$ ($0\leq i \leq 3$) in $\mathcal{F}_f$.
Furthermore, $p_1p'_1$ sees $f_1f'_1$.
Analogously, $p_2$ sees all points $f'_i$ ($0\leq i \leq 3$) in $\mathcal{F}_f$, and $p_2p'_2$ sees $f_3f'_3$.
So, we can transform $\mathcal{F}_s$ into $\mathcal{F}_f$ in five steps as in \cref{fig:5flips} if $p_1p'_1$ sees $f_1f'_1$ or $p_2p'_2$ sees $f_3f'_3$. 

\subsubsection{The vertex separators $\mathcal{I}$}

We describe now the placement of the vertex separators of a vertex gadget $\mathcal{V}$.
Since they are constructed in the same way on both sides, we describe how the left vertex separator $\mathcal{I}_l$ is placed.
Let $\mathcal{E}_{a}$ and $\mathcal{E}_{b}$ be the leftmost edge gadgets above and below the vertex frame $\mathcal{H}$.

We place the $2k+2$ vertical edges of the vertex separator $\mathcal{I}_l$ one unit to the left of $\mathcal{H}$.
The topmost points of these edges of $\mathcal{I}_l$ are one unit below the middle between $\mathcal{H}$ and $\mathcal{E}_a$, and the bottommost points of these edges of $\mathcal{I}_l$ are one unit above the middle between $\mathcal{H}$ and $\mathcal{E}_b$.
We place a set of $2k+2$ horizontal edges above and below the vertical edges.
The leftmost points of each of these edges is to the left of the vertical edges, whereas the rightmost points are close to $\mathcal{E}_a$ or $\mathcal{E}_b$, respectively. 
%Each of these edges goes from the leftmost edge of the horizontal edges until the it is close to $\mathcal{E}_a$ or $\mathcal{E}_b$, respectively. \todo{Precise placement, mention slopes}

Note that, by construction, every edge between a point of $\mathcal{I}$ and a point of a flip structure $\mathcal{F}_s$ either crosses a separator $\mathcal{S}$ of an edge gadget or the middle edges of a vertex gadget.

\begin{lemma}
Let $u_1v_1$ be the top or bottom edge of the frame of a vertex gadget $\mathcal{V}_1$ and let $u_2v_2$ be the top or bottom edge of the frame of a vertex gadget $\mathcal{V}_2$.
Then it takes at least $k$ flips such that $u_1v_1$ and $u_2v_2$ see each other.
\end{lemma}

\begin{proof}
Let $u_1$ be the leftmost point of $e_1$ and let $u_2$ be the leftmost point of $e_2$.
    Since $e_1$ and $e_2$ define a convex $4$-gon, we get $u_1u_2$ and $v_1v_2$ by flipping $e_1$ with $e_2$. 

    Assume $u_2$ is to the right of $u_1$.
    Then $u_1u_2$ crosses $2k+2$ edges of $\mathcal{I}$ of $\mathcal{V}_2$ by construction.
    By Lemma~\ref{lem:2k_crossings}, we need $k$ flips to obtain a configuration where we can flip $e_1$ with $e_2$.
\end{proof}

\subsection{Seeing edge with point not in edge gadget}
So far, we have shown that we need at least $8$ flips to transform $\mathcal{F}_s$ into $\mathcal{F}_f$ if the edge that sees an edge of $\mathcal{F}_s$ connects two points of the corresponding edge gadget.
Now, we study the case where there is an edge $p_1p_2$ that sees an edge $f_if'_i$ of an initial edge gadget $\mathcal{E}$ and $p_1$ is not a point of $\mathcal{E}$.

\begin{lemma}\label{lem:2k_crossings}
    Let $p_1$ be a point of a blocking structure $\mathcal{B}$ or of the flip structure $\mathcal{F}_s$ of an edge gadget $\mathcal{E}$, and let $p_2$ be a point in $\ms$ not in $\mathcal{E}$ or an adjacent vertex gadget.
    Then $p_1p_2$ crosses at least $2k+2$ edges in $\ms$.
\end{lemma}

\begin{proof}
    %Let $\mathcal{E}_1, \mathcal{E}_2$ be the corresponding edge gadgets of $p_1, p_2$.
    Let $\mathcal{V}_a$ and $\mathcal{V}_b$ be the above and below vertex gadgets incident to $\mathcal{E}$

   We prove that the line $\ell$ spanning $p_1$ and $p_2$ crosses the separator of $\mathcal{E}$, the middle edges  of $\mathcal{V}_a$ or $\mathcal{V}_b$, or the vertex separators of $\mathcal{V}_a$ or $\mathcal{V}_b$.

    Proof of claim:
    Without loss of generality, assume that $p_2$ is higher and to the right of $p_1$.
    Note that, by construction, any edge between any of the topmost points of the vertex separator $\mathcal{I}$ of $\mathcal{V}_a$ and a point of $\mathcal{E}$ crosses all the middle edges of $\mathcal{V}_a$.
    Also, note that $p_1$ is higher than the bottommost point of $\mathcal{S}$.
    Hence, if $\ell$ does not cross $\mathcal{S}$, then $\mathcal{S}$ is below $\ell$.
    On the other hand, the bottommost points of $\mathcal{I}$ of $\mathcal{V}_a$ are lower than the topmost points of $\mathcal{S}$.
    Hence, if $\ell$ does not cross $\mathcal{S}$ or $\mathcal{I}$, then $\mathcal{S}$ and $\mathcal{I}$ are below $p_1p_2$.
    It follows that, in this case, $\ell$ crosses the middle edges of $\mathcal{V}_a$ if $\ell$ does not cross $\mathcal{S}$ or $\mathcal{I}$.

So the only case left to study is if $\mathcal{I}$ and $\mathcal{S}$ are not between $p_1$ and $p_2$.
If $p_1$ and $p_2$ are on different sides of $\mathcal{V}_a$, then the line segment $p_1p_2$ crosses all the middle edges of $\mathcal{V}_a$.
If $p_1$ and $p_2$ are on the same side of $\mathcal{V}_a$, then $p_1p_2$ crosses $\mathcal{S}$ since $p_2$ is at most as high as the topmost point of $\mathcal{E}_1$ and $p_1$ is below the topmost point of $\mathcal{E}_1$.

Each of $\mathcal{S}_r, \mathcal{I}_r$ and the middle edges of $\mathcal{V}_a$ consists of $2k+2$ edges.
Hence, the line segment $p_1p_2$ crosses at least $2k+2$ edges.
\end{proof}

%\section{Conclusion and further research}
%\todo[inline]{Maybe add stuff here, also mention}

%We remark that our construction may contain several points on a line. 

\end{document}